%% file: main.tex
\documentclass[11pt]{amsart}

\usepackage{tikz-cd} \usepackage{adjustbox}

\input{math_symbs_header.tex}

\usepackage[top=1.5in,bottom=1.5in,left=1.5in,right=1.5in]{geometry}
\usepackage{braket}

\title[Satisfiability and boolean constraint system algebras]{Satisfiability
problems and algebras of boolean constraint system games}

\author[C. Paddock]{Connor Paddock$^{1}$} \author[W. Slofstra]{William
Slofstra$^{2,3}$}

\address[1]{Department of Mathematics and Statistics, University of Ottawa, Canada}
\address[2]{Institute for Quantum Computing, University of Waterloo, Canada}
\address[3]{Department of Pure Mathematics, University of Waterloo,
Canada}

\email{weslofst@uwaterloo.ca}

\DeclareMathOperator{\TRUE}{TRUE}
\DeclareMathOperator{\FALSE}{FALSE}
\DeclareMathOperator{\SAT}{SAT}
\DeclareMathOperator{\LIN}{LIN}
\newcommand{\wtd}{\widetilde}
\DeclareMathOperator{\oneink}{1-IN-k}

\begin{document}

\maketitle

\begin{abstract}
Mermin and Peres showed that there are boolean constraint systems (BCSs) which
are not satisfiable, but which are satisfiable with quantum observables. This
has led to a burgeoning theory of quantum satisfiability for constraint
systems, connected to nonlocal games and quantum contextuality. In this theory,
different types of quantum satisfying assignments can be understood as
representations of the BCS algebra of the system. This theory is closely
related to the theory of synchronous games and algebras, and every synchronous
algebra is a BCS algebra and vice-versa. The purpose of this paper is to
further develop the role of BCS algebras in this theory, and tie up some loose
ends: We give a new presentation of BCS algebras in terms of joint spectral
projections, and show that it is equivalent to the standard definition. We
construct a constraint system which is $C^*$-satisfiable but not tracially
satisfiable. We show that certain reductions between constraint systems lead to
$*$-homomorphisms between the BCS algebras of the systems, and use this to
streamline and strengthen several results of Atserias, Kolaitis, and Severini
on analogues of Schaefer's dichotomy theorem. In particular, we show that the
question of whether or not there is a non-hyperlinear group is linked to
dichotomy theorems for $\mcR^{\mcU}$-satisfiability. 
\end{abstract}

\section{Introduction}

Suppose $B$ is a boolean\footnote{In the nonlocal game literature, these have
previously been referred to as binary constraint systems. We use boolean instead
of binary to avoid confusion with $2$-ary relations. Fortunately both words lead
to same acronym.} constraint system (BCS), meaning a finite set of constraints
on a finite set of boolean variables. The BCS nonlocal game associated to $B$ is
a game with two players, who are unable to communicate while the game is in
progress. One player is given a constraint chosen from the system at random, and
must reply with a satisfying assignment to the variables in the given
constraint. The other player is given a variable from the constraint, and must
reply with an assignment to the variable. The players win if they assign the
same value to the chosen variable. With shared randomness, the players can play
this game perfectly if and only if $B$ is satisfiable, meaning that there is an
assignment to the variables of $B$ satisfying all the constraints. However,
there are constraint systems where the players can play perfectly by sharing an
entangled quantum state, even though the constraint system is not satisfiable.
The first examples of such systems were constructed by Mermin and Peres
\cite{Mer90,Per90}, although they were looking at quantum contextuality
scenarios rather than nonlocal games. The first game reformulation of Mermin and
Peres' examples was given by Aravind \cite{Ara02}, and the general class of BCS
games was introduced by Cleve and Mittal \cite{CM14}.

Cleve and Mittal show that the BCS game associated to $B$ has a perfect
(finite-dimensional) quantum strategy if and only if $B$ is matrix-satisfiable,
in the sense that every variable can be assigned a unitary matrix with $\pm
1$-eigenvalues, such that for all constraints, the matrices assigned to the
variables in that constraint are jointly diagonalizable, and (when $a \in \Z_2$
is encoded as the eigenvalue $(-1)^a \in \R$) the joint eigenvalues satisfy the
constraint. Implicit in their result is that $B$ is matrix-satisfiable if and
only if a certain $*$-algebra $\msA(B)$, called the BCS algebra associated to
$B$,  has a finite-dimensional $*$-representation. This algebra appeared
explicitly in \cite{Ji13}, and more recently in \cite{Pad22}. The standard 
definition of $\msA(B)$ in these works (both implicit and explicit) uses
multilinear polynomials, and can be inconvenient to work with. In
\Cref{sec:bcs_algebras}, we give an alternative definition of $\msA(B)$ in
terms of projections, and show that this presentation is equivalent to the
standard presentation.

The relationship between BCS games and BCS algebras parallels the situation for
synchronous games, in that a synchronous game has a perfect quantum strategy if
and only if the associated synchronous algebra has a finite-dimensional
$*$-representation. The concept of a synchronous game arose from the study of
quantum chromatic numbers in \cite{PSSTW16}, and synchronous algebras were
first defined explicitly by Helton, Meyer, Paulsen, and Satriano in
\cite{HMPS19}. It turns out that synchronous games and BCS games are closely
related. Every BCS game can be turned into a synchronous game, such that the
modified game has a perfect strategy if and only if the original game has a
perfect strategy. Furthermore, the synchronous algebra of this modified game is
isomorphic to the BCS algebra of the underlying constraint system. This
connection between BCS games and synchronous games was first shown by
Kim-Schafhauser-Paulsen \cite{KPS18} and Goldberg \cite{Gol21} for linear
constraint systems; the general case (which is very similar) is explained in
\Cref{sec:bcs_algebras}. Going in the opposite direction, Fritz shows that
the class of synchronous algebras is equal to the class of finite colimits of
finite-dimensional abelian $C^*$-algebras \cite{Fri18}. It is easy to see that
this latter class is equal to the class of BCS algebras, so in particular every
synchronous algebra is isomorphic to a BCS algebra. A direct proof of this fact
using constraint systems is given in \cite{Pad22}; the construction is
sketched in \Cref{sec:bcs_algebras} for the convenience of the reader.

Much of the interest in BCS games and synchronous games has been in using them
to separate different models of entanglement. Mermin and Peres' examples show
that finite-dimensional quantum entanglement is more powerful than shared
randomness. In \cite{Slof17}, the second author constructs examples of linear
constraint systems which are not matrix-satisfiable, but which are
$\mcR^{\mcU}$-satisfiable, in the sense that there is a homomorphism $\msA(B)
\to \mcR^{\mcU}$, where $\mcR^{\mcU}$ is an ultrapower of the hyperfinite
II$_1$ factor $\mcR$. The resulting BCS games cannot be played perfectly with a
finite-dimensional quantum state, but can be played perfectly with a limit of
finite-dimensional states. The celebrated MIP*=RE theorem of Ji, Natarajan,
Vidick, Wright, and Yuen shows (via synchronous games) that there are
constraint systems which are not $\mcR^{\mcU}$-satisfiable, but are tracially
satisfiable, in the sense that $\msA(B)$ has a tracial state \cite{JNVWY22}. In
this case, the BCS games can be played perfectly in the commuting operator
model of entanglement, but cannot be played perfectly with any limit of
finite-dimensional states.  Helton, Meyer, Paulsen, and Satriano show (again
via synchronous games) that there are constraint systems which are
algebraically-satisfiable, in the sense that $\msA(B)$ is non-trivial, but not
$C^*$-satisfiable, meaning that $\msA(B)$ does not admit a state \cite{HMPS19}.
In Section \ref{sec:separations}, we complete this picture by showing that
there is a constraint system which is $C^*$-satisfiable, but not
tracially-satisfiable. This answers a question of \cite{Har21}. While the
existence of such constraint systems does not have a nonlocal game
interpretation that we know of, it does imply that there are contextuality
scenarios where the algebra generated by the operators does not admit a tracial
state. 

Deciding whether a constraint system is satisfiable is a well-known NP-complete
problem. However, for some classes of constraint systems, such as linear
systems, 2SAT systems, Horn systems, and so on, satisfiability can be decided
in polynomial time. Schaefer's dichotomy theorem states that for any finite
list of constraints (called a constraint language), the satisfiability problem
for systems built out of these constraints is either decidable in polynomial
time, or NP-complete \cite{Sch78}. The languages for which satisfiability can be decided in
polynomial time are called Schaefer classes.  The key idea in the proof of
Schaefer's theorem is to consider what other relations can be defined from a
given constraint language. For instance, it's well-known that any relation is
definable from 3SAT relations, and consequently 3SAT is NP-complete. Ji showed
that some reductions between constraint languages (for instance, from k-SAT to
3-SAT) also hold for BCS games, and that satisfiability and
matrix-satisfiability are the same for 2SAT and Horn systems \cite{Ji13}. This
subject was then investigated systematically by Atserias, Kolaitis, and
Severini \cite{AKS17}. They show that satisfiability, matrix-satisfiability,
and $C^*$-satisfiability are the same for all the Schaefer classes except
linear systems. They observe that linear systems are definable from any
non-Schaefer classes. Since matrix and $C^*$-satisfiability are undecidable for
linear systems by \cite{Slof19,Slof17}, they arrive at an analogue of Schaefer's
dichotomy theorem for matrix-satisfiability and $C^*$-satisfiability: for every
constraint language, matrix-satisfiability and $C^*$-satisfiability are either
decidable in polynomial time, or undecidable. In addition, there is a dichotomy
for separations of different types of satisfiability: for all constraint
languages, satisfiability, matrix-satisfiability, and $C^*$-satisfiability are
either all equivalent, or all different. 

In the final section of this paper, we show that when a constraint language
$\mcL$ is definable from another language $\mcL'$, then for every constraint
system $B$ over $\mcL$, there is a constraint system $B'$ over $\mcL'$ such
that there are $*$-homomorphisms $\msA(B) \to \msA(B')$ and $\msA(B') \to
\msA(B)$. As a result, all the types of satisfiability discussed above are
equivalent for $B$ and $B'$. We use this fact to streamline and extend
Atserias, Kolaitis, and Severini's results. In particular, we show that
tracial\nobreakdash-, $C^*$\nobreakdash-, and algebraic-satisfiability are all
either decidable in polynomial time, or coRE-complete. We point out that while
matrix-satisfiability and $\mcR^{\mcU}$-satisfiability are either decidable in
polynomial time or undecidable, the exact complexity in the latter case is not
known. We refine Atserias, Kolaitis, and Severini's separation dichotomy
theorem to show that for all constraint languages, either satisfiability is the
same as algebraic-satisfiability, or satisfiability, matrix-satisfiability, and
$\mcR^{\mcU}$-satisfiability are all different. The sticking point to including
tracial-satisfiability in this list is whether or not there is a linear system
which is tracially satisfiable but not $\mcR^{\mcU}$-satisfiable, a problem
which is equivalent to the existence of a non-hyperlinear finitely-presented
group. 

On the other hand, tracial-satisfiability, $C^*$-satisfiability, and
algebraic-satisfiability are all equivalent for linear systems by \cite{CLS17}. 
This fact, combined with the existence of constraint systems which are
algebraically satisfiable but not tracially satisfiable, implies that it is not
possible to embed every BCS algebra in the BCS algebra of a linear system game.
This is an obstacle for one method of constructing a non-hyperlinear group that
tends to come up when discussing potential consequences of MIP*=RE.
Combined with a lemma from \cite{Sch78}, this also shows that for all
constraint languages, tracial-satisfiability, $C^*$-satisfiability, and
algebraic-satisfiability are either all equivalent, or all different. 
This is discussed further in \Cref{sec:define}. 

\section*{Acknowledgements}

We thank Albert Atserias, Richard Cleve, Tobias Fritz, Sam Harris, and Vern Paulsen for helpful discussions.
This work was supported by NSERC DG 2018-03968 and an Alfred P.  Sloan Research
Fellowship. 

\section{Preliminaries and notation}

\subsection{Finitely-presented $*$-algebras}
For a set $X$, let $\C^*\langle X\rangle$ be the free unital complex
$*$-algebra generated by $X$. If $R \subseteq \C^*\langle X \rangle$, let
$\C^* \langle X : R \rangle$ denote the quotient of $\C^*\langle X \rangle$ 
by the two-sided ideal $\langle \langle R \rangle \rangle$ generated by $R$.
If $X$ and $R$ are finite, then $\C^* \langle X : R \rangle$ is said to be a
finitely-presented $*$-algebra.

If $W$ is a complex vector space, then $\End(W)$ will denote the space of
linear operators from $W$ to itself. We use $\Id_W$ for the identity operator
on $W$. A representation of a $*$-algebra $\msA$
is an algebra homomorphism $\phi : \msA \to \End(W)$ for some vector space $W$.
A subrepresentation is a subspace $K \subseteq W$ such that $\phi(a) K
\subseteq K$ for all $a \in \msA$, and a representation is irreducible if it
has no subrepresentations. A $*$-representation of $\msA$ is a $*$-homomorphism
$\phi : \msA \to \mcB(H)$, where $H$ is a Hilbert space, and $\mcB(H)$ is the
$*$-algebra of bounded operators on $H$. 

If $\msA = \C^*\langle X : R \rangle$ is a presentation of a $*$-algebra, and
$\msB$ is another $*$-algebra, then $*$-homomorphisms $\msA \to \msB$
correspond to $*$-homomorphisms $\phi : \C^*\langle X \rangle \to \msB$ such
that $\phi(r) = 0$ for all $r \in R$. Hence a $*$-representation of $\msA$ is
an assignment of operators to the elements of $X$, such that the operators
satisfy the defining $*$-relations in $R$, and we often work with representations
in these terms.

An element $x$ of a $*$-algebra algebra $\msA$ is said to be \textbf{positive},
written $x \geq 0$, if $x = \sum_{i=1}^k s_i^* s_i$ for some $k \geq 1$ and
$s_1,\ldots,s_k \in \msA$. The algebra $\msA$ is said to be a
\textbf{semi-pre-$C^*$-algebra} if for all $x \in \msA$, there is a scalar $\lambda \geq
0$ such that $x^* x \leq \lambda\cdot 1$ \cite{Oza13}. All the $*$-algebras we work
with will be semi-pre-$C^*$-algebras.  A \textbf{state} on a semi-pre-$C^*$-algebra
$\msA$ is a linear functional $f : \msA \to \C$ such that $f(x^*) =
\overline{f(x)}$ for all $x \in \msA$, $f(x) \geq 0$ for all $x \geq 0$, and
$f(1) = 1$. If $\phi : \msA \to \mcB(H)$ is a $*$-representation of $\msA$, and
$\ket{v} \in H$ is a unit vector, then $x \mapsto \braket{v|\phi(x)|v}$ is a
state.  Conversely, if $f$ is a state then by the GNS representation theorem,
there is a $*$-representation $\phi : \msA \to \mcB(H)$ and a unit vector
$\ket{v} \in H$ such that $f(x) = \braket{v|\phi(x)|v}$ for all $x \in \msA$.
Hence a semi-pre-$C^*$-algebra $\msA$ has a state if and only if it has a
$*$-representation.  A state $f$ on $\msA$ is \textbf{tracial} if $f(ab) =
f(ba)$.

\subsection{The joint spectrum and representations of $\C\Z_2^k$}

Consider the finitely presented group 
\begin{equation*} 
    \Z_2^k=\langle z_1,\ldots,z_k:z_i^2=1,1\leq i\leq k,z_iz_j=z_jz_i,1\leq i\neq j\leq k\rangle.
\end{equation*} 
The group algebra $\C \Z_2^k$ is the $*$-algebra generated by $z_1,\ldots,z_k$,
with defining relations from the group presentation of $\Z_2^k$ above, along
with the relations $z_i^* z_i = z_i z_i^* = 1$ for all $1 \leq i \leq k$. Hence
a $*$-representation of $\C \Z_2^k$ is a collection of unitary operators $Z_1,\ldots,Z_k$
such that $Z_i^2=\Id$ and $Z_i Z_j = Z_j Z_i$ for all $1 \leq i, j \leq k$. The
irreducible representations $\lambda$ of $\C \Z_2^k$ are one-dimensional, and
are determined via the vectors $v = (\lambda(z_1),\ldots,\lambda(z_k)) \in \{\pm 1\}^k$.
Conversely, for any vector $v \in \{\pm 1\}^k$ there is a representation $\lambda_v : \Z_2^k
\to \C$ with $\lambda_v(x_i) = v_i$, so $\Z_2^k$ has $2^k$ irreducible representations (up to 
isomorphism). If $v \in \{\pm 1\}^k$, we let 
\begin{equation*}
    \Pi_v = \prod_{i=1}^k \tfrac{1}{2} (1 + v_i z_i) = \frac{1}{2^k} \sum_{x \in \Z_2^k} \lambda_v(x) x 
\end{equation*}
be the central projection in $\C \Z_2^k$ corresponding to $\lambda_v$. These
projections satisfy the identities $\Pi_v^* = \Pi_v = \Pi_v^2$ for all $v \in
\{\pm 1\}^k$, $\Pi_v \Pi_{v'} = 0$ for all $v \neq v' \in \{\pm 1\}^k$, and
$\sum_{v \in \{\pm 1\}^k} \Pi_v = 1$ in $\C \Z_2^k$. Hence if $\psi : \C \Z_2^k
\to \End(W)$ is a representation, the operators $\{\psi(\Pi_v)\}_{v \in
\{\pm 1\}^k}$ form a complete orthogonal family of projections. In particular, 
\begin{equation*}
    W = \bigoplus_{v \in \{\pm 1\}^k} \psi(\Pi_v) W
\end{equation*}
as an direct sum, and the subspaces are orthogonal if $W$ is a
$*$-representation.  If $x \in \C \Z_2^k$ and $w \in W_v := \psi(\Pi_v) W$,
then $\psi(x) w = \lambda_v(x) w$. Hence $\psi(x) = \oplus_{v \in \{\pm
1\}^k} \lambda_v(x) \Id_{W_v}$ for all $x \in \C \Z_2^k$, so $\psi$ is diagonal
with respect to this subspace decomposition.  This leads to the following
standard definition:

\begin{definition} 
    Let $\psi:\C\Z_2 ^k\to \End(W)$ be a representation of $\C\Z_2^k$ on a vector
    space $W$. The \textbf{joint spectrum} of $\psi$ is the set
    \begin{equation*}\mcJ_\psi=\{v\in \{\pm1\}^k:\psi(\Pi_v)\neq 0\}.
    \end{equation*} 
\end{definition}
In other words, $\mcJ_{\psi}$ is the set of vectors $v \in \{\pm 1\}^k$ for
which the subspace $W_v$ is non-zero.

\subsection{Boolean constraint systems}

To match with conventions from the previous section, in this paper we represent
$\Z_2$ in multiplicative form as $\{\pm 1\}$, rather than $\{0,1\}$. We also
use this convention for boolean truth values, meaning we think of $-1$ as
$\TRUE$ and $1$ as $\FALSE$. A \textbf{boolean relation of arity $k>0$} is a
subset of $\{\pm 1\}^k$. The \textbf{indicator function} of a relation $R
\subset \{\pm 1 \}^k$ is the function $f_R : \{\pm 1\}^k \to \{\pm 1\}$ sending
$x \mapsto -1$ if $x \in R$, and $x \mapsto 1$ otherwise.  Given a set of
variables $X=\{x_1,\ldots,x_n\}$, a \textbf{constraint} $C$ on $X$ is a pair
$(S,R)$, where the \textbf{scope} $S = (s_1,\ldots,s_k)$ is a sequence of
length $k \geq 1$ over $X \cup \{\pm 1\}$, and $R$ is a $k$-ary relation. Note
that this definition of scope allows us to substitute constants for variables
in the relation $R$.  A \textbf{boolean constraint system (BCS)} is a pair $(X,
\{C_i\}_{i=1}^m)$, where $X$ is a finite set of variables, and
$\{C_i\}_{i=1}^m$ is a finite set of constraints on $X$.

For practical reasons, we will often write relations and constraints informally
in the standard short-hand, using $\vee$ for logical OR and $\wedge$ for
logical AND. For instance, $x_1 \vee x_2 \vee x_3 = \TRUE$ could refer to the
relation $R = \{\pm 1\}^3 \setminus \{(1,1,1)\}$, or to the constraint
$((x_1,x_2,x_3), R)$. 
Also, if $S$ is a scope, we
abuse notation slightly and use $X \cap S$ to refer to the set of variables
listed in $S$.  Note a peculiarity of using multiplicative notation for $\Z_2$
is that the XOR $x_1 \oplus x_2$ is written as the product $x_1 x_2$, so for
instance $\{x_1 x_2 x_3 = -1, x_2 x_3 = 1, x_1 x_2 = 1\}$ is actually a linear
system, despite initial appearances. Similarly, $-x_i$ is the logical negation
of the variable $x_i$.

An \textbf{assignment} to a set of variables $X$ is a function $\phi : X \to
\{\pm 1\}$.  If $S = (s_1,\ldots,s_k)$ is a sequence over $X \cup \{\pm 1\}$,
we set $\phi(S) = (\phi(s_1),\ldots,\phi(s_k)) \in \{\pm 1\}^k$, where we
extend $\phi$ to $\{\pm 1\}$ as the identity function. If $(X, \{(S_i,
R_i)\}_{i=1}^m)$ is a BCS, then an assignment $\phi$ to $X$ is a
\textbf{satisfying assignment} if $\phi(S_i) \in R_i$ for all $1 \leq i \leq
m$, or equivalently if $f_{R_i}(\phi(S_i)) = -1$ for all $1 \leq i \leq m$.
An assignment which is not a satisfying assignment will be called a
\textbf{non-satisfying assignment}.  A BCS is said to be \textbf{satisfiable}
if it has a satisfying assignment. 

A \textbf{boolean constraint language} $\mcL$ is a collection of relations with
possibly different arity's.  We say that a BCS $B$ is a \textbf{BCS over
$\mcL$} if every relation in $B$ belongs to $\mcL$. Constraint languages allow
us to talk about constraint systems where the relations are of a certain form.
For example, a 3SAT instance is a BCS over the constraint language containing
the relations $(\pm x_1) \vee (\pm x_2) \vee (\pm x_3) = \TRUE$. 
Other examples of constraint languages will be given in \Cref{sec:define}.

\section{Boolean constraint system algebras and games}\label{sec:bcs_algebras}

To define boolean constraint system algebras, we first extend the definition
of a boolean constraint system slightly:
\begin{defn}
    A \textbf{boolean constraint system (BCS) with contexts} is a tuple $(X,
    \{(U_i,V_i)\}_{i=1}^\ell)$, where $X$ is a finite set of variables, and
    $(U_i,V_i)$ is a constraint system on variables $U_i \subseteq X$ for all
    $1 \leq i \leq \ell$.  The sets $U_i$, $1 \leq i \leq \ell$ are called the
    \textbf{contexts} of the system. 

    If $\mcL$ is a constraint language, then a \textbf{BCS with contexts over
    $\mcL$} is a BCS with contexts $(X, \{(U_i,V_i)\}_{i=1}^\ell)$ in which
    $(U_i,V_i)$ is a BCS over $\mcL$ for all $1 \leq i \leq \ell$.
\end{defn}
Intuitively, the idea behind this definition is that variables are grouped into
contexts, and constraints can only be placed on variables in the same context.
Given a BCS $B = (X,V)$, we can always add contexts to make it into a BCS with
contexts. Typically there is more than one way that this can be done.  For
instance, we could group all the variables together into a single context to
get $(X,\{(X,V)\})$. At the other end of the spectrum, we could add a separate context
for each constraint, containing only the variables in that constraint. In this
case, if $V = \{(S_i,R_i)\}_{i=1}^{\ell}$, then the BCS with contexts is
$(X, \{(U_i,V_i)\})$, where $U_i \subseteq X$ is the set of variables appearing
in $S_i$, and $V_i = \{(S_i,R_i)\}$. We use this option as the default option
when thinking of a BCS as a BCS with contexts. 
\begin{definition}\label{defn:bcs_algebra}
    Let $B = (X, \{U_i,V_i\}_{i=1}^\ell)$ be a BCS with contexts. Let
    $\msA_{con}(B)$ be the finitely-presented $*$-algebra generated by $X$ and
    subject to the relations
    \begin{enumerate}
        \item $x^2=1$ and $x^*=x$ for all $x \in X$, and
        \item $xy = yx$ for all $x,y \in U_i$, $1 \leq i \leq m$.
    \end{enumerate}
    For $1 \leq i \leq \ell$ and $\phi$ an assignment to $U_i$, let 
    $\Pi_{U_i,\phi}$ denote the projection
    \begin{equation*}
        \prod_{x \in U_i} \tfrac{1}{2} (1+\phi(x) x)
    \end{equation*}
    in $\msA_{con}(B)$. The \textbf{boolean constraint system algebra} $\msA(B)$ is the
    quotient of $\msA_{con}(B)$ by the relations
    \begin{enumerate} 
        \setcounter{enumi}{2}
        \item $\Pi_{U_i,\phi} = 0$ for all $1 \leq i \leq \ell$ and non-satisfying assignments $\phi$ for $(U_i,V_i)$. 
    \end{enumerate}
\end{definition}
The relations in (1) imply that $\msA(B)$ is a semi-pre-$C^*$-algebra. 
As mentioned in the preliminaries, a $*$-representation of $\msA_{con}(B)$ is an
assignment $x \mapsto \psi(x)$ of operators to every variable $x \in B$, such
that $\psi(x)^* = \psi(x)$ and $\psi(x)^2 = 1$ for all $x \in X$, and $\psi(x)
\psi(y) = \psi(y)\psi(x)$ for all $x,y \in U_i$ and $1 \leq i \leq \ell$. If
$\psi$ is a $*$-representation of $\msA_{con}(B)$, then for any $1 \leq i \leq
\ell$, the operators $\psi(x)$, $x \in U_i$ form a $*$-representation of
$\Z_2^{U_i}$. The $*$-subalgebra of $\msA_{con}(B)$ generated by $U_i$ is
isomorphic to the group algebra of $\Z_2^{U_i}$. An assignment $\phi$ to $U_i$
is equivalent to an irreducible representation of $\Z_2^{U_i}$, and
$\Pi_{U_i,\phi}$ is the central projection in $\C \Z_2^{U_i}$ corresponding to
$\phi$. Hence a representation $\psi$ of $\msA_{con}(B)$ induces a representation
of $\msA(B)$ if and only if for all $1 \leq i \leq \ell$, the joint spectrum
$\mcJ_{\psi_{U_i}}$ is contained in the set of satisfying assignments for
$(U_i,V_i)$. In particular, a one-dimensional $*$-representation of $\msA(B)$
is the same thing as a satisfying assignment for $B$. 

More generally, if $B = (X, (X,V))$ is a BCS with all the variables in one context, then
$\msA(B)$ is commutative, and a $*$-representation of $\msA(B)$ is a direct sum
of satisfying assignments for $(X,V)$. However, with more contexts it's
possible to have $*$-representations even when there are no one-dimensional
representations. The Mermin-Peres magic square is a famous example:
\begin{example}\label{ex:mp_square}
    The Mermin-Peres magic square is the constraint system $B$ over $X =
    \{x_1,\ldots,x_9\}$ with constraints $x_1 x_2 x_3 = 1$, $x_4 x_5 x_6 = 1$,
    $x_7 x_8 x_9 = 1$, $x_1 x_4 x_7 = -1$, $x_2 x_5 x_8 = -1$, and $x_3 x_6 x_9
    = -1$.  These constraints arise from putting the variables $x_1,\ldots,x_9$
    in a $3 \times 3$ grid 
    \begin{center}
    \begin{tabular}{|c|c|c|}
        \hline
        $x_1$ & $x_2$ & $x_3$ \\
        \hline
        $x_4$ & $x_5$ & $x_6$ \\
        \hline
        $x_7$ & $x_8$ & $x_9$ \\
        \hline
    \end{tabular}\ \ ,
    \end{center}
    and requiring that the row products are $1$ and the column products are
    $-1$. If we think of this BCS as a BCS with contexts in the default way
    (adding a context for each constraint), then $\msA(B)$ does not have
    a one-dimensional $*$-representation, but does have a $*$-representation
    in dimension $4$ \cite{Mer90}.
\end{example}

Note that if we are working over the constraint language of all relations, then
it is somewhat redundant to explicitly specify contexts. Indeed, given a BCS
with contexts $B = (X, \{(U_i,V_i)\}_{i=1}^{\ell})$, let $C_i$ be the
constraint $(S_i, R_i)$, where $S_i$ is an enumeration of $U_i$, and $R_i$ is
the set of satisfying assignments to $(U_i,V_i)$. Let $B' = (X,
\{C_i\}_{i=1}^{\ell})$, an ordinary BCS without contexts. If we regard $B'$ as
a BCS with contexts in the default way mentioned above, then $\msA(B) =
\msA(B')$, so we can always assume that the contexts are defined implicitly
from the relations. However, as the following example shows, being able to
explicitly specify contexts is convenient when working over more restrictive
constraint languages: 
\begin{example}\label{ex:commute1}
    Let $X = \{x_1,x_2,x_3,x_4\}$, and let $B$ be the 3SAT instance $(X, V)$,
    where $V = \{x_1 \vee x_2 \vee x_3=\TRUE, x_2 \vee x_3 \vee x_4=\TRUE\}$.
    Then $x_1$ and $x_4$ do not commute in $\msA(B)$.  If we want them to
    commute, we can instead use the BCS with contexts $B' = (X, \{(X, V)\})$.
    We also have $\msA(B') = \msA(B'')$ where $B'' = (X, \{(x_1 \vee x_2 \vee
    x_3) \wedge (x_2 \vee x_3 \vee x_4)\}$ is a BCS without contexts. However,
    $B''$ is not a 3SAT instance.
\end{example}

We can also consider other presentations of $\msA(B)$. Recall that if $S$ is the scope of a
constraint over variables $X$, then $X \cap S$ refers to the set of variables
in $X$.
\begin{lemma}\label{lem:presperconstraint}
    Suppose $B = (X, \{(U_i,V_i)\}_{i=1}^{\ell})$ is a BCS with contexts. Given a 
    constraint $C = (S,R) \in V_i$, $1 \leq i \leq \ell$ and an assignment $\phi$
    to $X \cap S$, let 
    \begin{equation*}
        \Pi_{C,\phi} = \prod_{x \in S \cap X} \tfrac{1}{2} (1 + \phi(x) x).
    \end{equation*}
    Then $\msA(B)$ is the quotient of $\msA_{con}(B)$ by the relations $\Pi_{C,\phi} = 0$
    for all $C = (S,R) \in V_i$, $1 \leq i \leq \ell$, and $\phi$ an assignment to
    $X \cap S$ such that $\phi(S) \not\in R$. 
\end{lemma}
\begin{proof}
    If $\phi$ is a non-satisfying assignment to $(U_i, V_i)$, then $\phi(S) \not\in R$
    for some $C = (S,R) \in V_i$, and $\Pi_{U_i,\phi}$
    is in the ideal generated by $\Pi_{C,\phi|_{X \cap S}}$.

    Conversely, suppose $C = (S,R) \in V_i$ and $\phi$ is an assignment to $X
    \cap S$ such that $\phi(S) \not\in R$. If $\wtd{\phi}$ is an assignment to 
    $U_i$ such that $\wtd{\phi}|_{X \cap S} = \phi$, then $\wtd{\phi}$ is a non-satisfying
    assignment to $(U_i, V_i)$. Hence
    \begin{equation*}
        \Pi_{C,\phi} = \sum_{\wtd{\phi}|_{X \cap S} = \phi} \Pi_{U_i,\wtd{\phi}}
    \end{equation*}
    is in the ideal generated by $\Pi_{U_i,\wtd{\phi}}$ for $\wtd{\phi}$ a non-satisfying
    assignment to $(U_i,V_i)$. 
\end{proof}
By thinking of the elements of $\C \Z_2^k$ as polynomials in the variables $z_1,\ldots,z_k$,
we can identify $\C \Z_2^k$ with the space of functions $\{\pm 1\}^k \to \C$. Specifically,
if $f : \{\pm 1\}^k \to \C$ is a function, then the corresponding element of $\C \Z_2^k$ is
\begin{equation*}
    P =\sum_{v\in \{\pm1\}^k} f(v)\prod_{i=1}^k \frac{\left(1+v_iz_i\right)}{2}.
\end{equation*}
When $f$ is the indicator function $f_R$ of a relation $R \subseteq \{\pm 1\}$,
we denote this element by $P_R$. If $C = (S,R)$ is a constraint over variables
$X$, the evaluation $P_R(S)$ of $P_R$ at $S = (s_1,\ldots,s_k)$ is the element
of $\C \Z_2^{X \cap S}$ that we get from replacing $x_i$ in $P_R$ with $s_i$
for all $1 \leq i \leq k$. It is not hard to see that 
\begin{equation*}
    P_R(S) = \sum_{\phi(S) \not\in R} \Pi_{C,\phi} - \sum_{\phi(S) \in R} \Pi_{C,\phi},
\end{equation*}
where the sums are over assignments $\phi$ to $X \cap S$, and $\Pi_{C,\phi}$ is
defined in \Cref{lem:presperconstraint}. Since $\sum_{\phi} \Pi_{C,\phi} = 1$, we conclude
that $1 + P_R(S) = 2 \sum_{\phi(S) \not\in R} \Pi_{C,\phi}$. This leads to the standard
presentation of $\msA(B)$ found in previous work:
\begin{cor}\label{cor:relationpres}
    Suppose $B = (X, \{(U_i,V_i)\}_{i=1}^{\ell})$ is a BCS with contexts. Then $\msA(B)$ is
    the quotient of $\msA_{con}(B)$ by the relations $P_R(S) = -1$ for all $C = (S,R) \in V_i$,
    $1 \leq i \leq \ell$.
\end{cor}
\begin{proof}
    The argument above shows that $1+P_R(S)$ is in the ideal generated by the projections $\Pi_{C,\phi}$
    with $\phi$ an assignment to $X \cap S$ such that $\phi(S) \not\in R$. Conversely, the projections
    $\Pi_{C,\phi}$ for $\phi$ an assignment to $X \cap S$ are orthogonal. Hence if $\phi(S) \not\in R$,
    then $\Pi_{C,\phi} = \tfrac{1}{2} \Pi_{C,\phi} (1 + P_R(S))$ is in the ideal generated by $1+P_R(S)$.
    The Corollary follows from \Cref{lem:presperconstraint}.
\end{proof}

\begin{example}\label{ex:AND}
    If $R$ is the AND relation $x \land y = \TRUE$, then $P_R(z_1,z_2) =
    \tfrac{1}{2}(1 + z_1 + z_2 - z_1z_2)$. If $B$ is the BCS with two variables $z_1,z_2$ and
    the single relation $R$, then $\msA(B) = \C^*\langle z_1,z_2 : z_i^*=z_i, z_i^2=1, i=1,2, P_R(z_1,z_2)=-1\rangle$.
\end{example}

\begin{example}\label{ex:zAND}
    Let $R'$ be the relation $z = x \land y$. Then $P_{R'}(x,y,z) =
    -\tfrac{1}{2} z (1 + x + y - xy)$, so the relation $P_{R'}(x,y,z) = -1$ 
    is equivalent to the relation $z = \tfrac{1}{2} (1 + x + y - xy) = P_R(x,y)$,
    where $R$ is the AND relation from the previous example.
\end{example}

\begin{example}\label{ex:linear}
    Suppose $Ax=b$ is an $m \times n$ linear system over $\Z_2$, written in
    additive notation (a.k.a.~the normal way of writing linear systems). Let
    $B = (X, \{(U_i,\{C_i\})\}_{i=1}^{m})$ be the corresponding BCS with
    contexts, so $X = \{x_1,\ldots,x_n\}$ is the set of variables, $U_i = \{x_j
    : A_{ij} \neq 0\}$ is the set of variables in equation $i$, and $C_i =
    (S_i,R_i)$ is the $i$th equation of the system, which written
    multiplicatively is $x_1^{A_{i1}} \cdots x_{n}^{A_{in}} = (-1)^{b_i}$. Then
    $P_{R_i}(S_i) = (-1)^{b_i+1} x_1^{A_{i1}} \cdots x_{n}^{A_{in}}$, and
    $\msA(B)$ is the finitely-presented $*$-algebra generated by
    $x_1,\ldots,x_n$ subject to the relations
    \begin{enumerate}[(a)]
        \item $x_i^2=1$ and $x_i^*=x_i$ for all $1 \leq i \leq n$,
        \item $x_j x_k = x_k x_j$ for all $x_j, x_k \in U_i$, $1 \leq i \leq m$, and
        \item $x_1^{A_{i1}} \cdots x_{n}^{A_{in}} = (-1)^{b_i}$ for all $1 \leq i \leq m$.
    \end{enumerate}
    Note that these relations are very close to the relations for a group
    algebra. The \emph{solution group} $\Gamma(A,b)$ of $Ax=b$, defined in
    \cite{CLS17}, is the finitely-presented group generated by $X \cup \{J\}$
    subject to the relations
    \begin{enumerate}[(1)]
        \item $J^2 = 1$ and $x_i J = J x_i$ for all $1 \leq i \leq n$,
        \item $x_i^2=1$ and $x_i^*=x_i$ for all $1 \leq i \leq n$,
        \item $x_j x_k = x_k x_j$ for all $x_j, x_k \in U_i$, $1 \leq i \leq m$, and
        \item $x_1^{A_{i1}} \cdots x_{n}^{A_{in}} = J^{b_i}$ for all $1 \leq m$.
    \end{enumerate}    
    Comparing the two relations for $\Gamma(A,b)$ to the relations for
    $\msA(B)$, we see that $\msA(B) = \C \Gamma(A,b) / \langle J=-1\rangle$. 
\end{example}

\begin{example}\label{ex:commute2}
    If we want to avoid explicitly specifying contexts, we can often add
    ancillary variables and relations. For instance, if we want to make two
    variables $x_1$ and $x_2$ of $B$ commute in $\msA(B)$, we could add a third
    variable $x_3$ and the linear relation $x_1 x_2 x_3 = 1$. Since this
    relation forces $x_3 = x_1 x_2$, the only effect of this relation (assuming
    that $x_3$ is not used in any other relations) is to force $x_1$ and $x_2$
    to commute. 
\end{example}

\subsection{Contextuality scenarios and nonlocal games}

Let $B = (X,\{(U_i,V_i)\}_{i=1}^{\ell})$ be a BCS with contexts, and suppose
$\psi$ is a $*$-representation of $\msA_{con}(B)$ on some Hilbert space $H$.
For any $1 \leq i \leq \ell$, the operators $\psi(x)$, $x \in  U_i$ are
jointly-measurable $\pm 1$-valued observables, with joint outcomes
corresponding to assignments $\phi$ to $X \cap S$. This type of measurement
scenario, in which a bunch of observables are grouped into contexts, 
observables in the same context are jointly measurable, observables from
different contexts are not necessarily jointly measurable, and observables can
belong to more than one context, is called a \textbf{contextuality scenario} (see, e.g. \cite{AFLS15}).
This physical interpretation of $*$-representations of $\msA_{con}(B)$ goes
back to the original papers of Mermin and Peres \cite{Mer90,Per90}.
If $\psi$ is a $*$-representation of $\msA(B)$, then the outcome of measuring
the operators $\psi(x)$, $x \in U_i$ (with respect to any state) is a
satisfying assignment to $(U_i,V_i)$.  If $B$ does not have a satisfying
assignment (or in other words, a one-dimensional $*$-representation), then the
operators $\psi(x), x \in X$ are said to be \textbf{contextual}, since the
measured value of $\psi(x)$ seems to depend on what context it is measured in.

Another physical interpretation of $*$-representations of $\msA(B)$ is provided
by the \textbf{BCS nonlocal game} $\mcG(B)$ associated to $B$. In this game, two players
(commonly called Alice and Bob) are each given an input $1 \leq i \leq \ell$,
and must respond with an assignment $\phi$ for $U_i$. If Alice
and Bob receive inputs $i$ and $j$ respectively, and respond with outputs
$\phi_A$ and $\phi_B$, then they win if $\phi_A$ and $\phi_B$ are satisfying
assignments for $(U_i, V_i)$ and $(U_j, V_j)$ respectively, and $\phi_A|_{U_i
\cap U_j} = \phi_B|_{U_i \cap U_j}$.  If this condition is not satisfied, then
they lose. The players are cooperating to win, and they know the rules and can
decide on a strategy ahead of time.  However, they are not able to communicate
once the game is in progress (so in particular, Alice does not know which
context Bob received, and vice-versa). 

There are different types of strategies Alice and Bob might use, depending on
the resource they have access to. A strategy is \textbf{classical}
if Alice and Bob have access to shared randomness; \textbf{quantum} if Alice
and Bob share a finite-dimensional bipartite entangled quantum state;
\textbf{quantum-approximable} if the strategy is a limit of quantum strategies;
and \textbf{commuting-operator} if Alice and Bob share a quantum state in a
possibly infinite-dimensional Hilbert space, and Alice's measurement operators
have to commute with Bob's operators (rather than acting on separate tensor
factors). We refer to \cite{CM14,LMPRSSTW20} for more background on quantum strategies.  

A \textbf{perfect strategy} for a nonlocal game is a strategy which wins on
every pair of inputs. If $B$ has a satisfying assignment, then $\mcG(B)$ has a
classical perfect strategy. Indeed, Alice and Bob can agree on a satisfying
assignment $\phi$ ahead of time, and respond with $\phi|_{U_i}$ on input $i$.
It turns out that $\mcG(B)$ has a perfect classical strategy if and only if $B$
has a satisfying assignment. The following theorem describes the relationship
between perfect strategies for $\mcG(B)$, and $*$-representations of the BCS
algebra $\msA(B)$. 
\begin{theorem}[\cite{CM14,PSSTW16,KPS18}]\label{thm:BCS_t-strats}
    Let $B$ be a BCS with contexts. Then:
\begin{enumerate}
	\item $\mcG(B)$ has a perfect classical strategy if and only if there is a $*$-homomorphism $\msA(B)\to \C$,
	\item $\mcG(B)$ has a perfect quantum strategy if and only if
there is a $*$-homomorphism $\msA(B)\to M_d(\C)$ for some $d\geq
1$,
    \item $\mcG(B)$ has a perfect quantum-approximable strategy if and only if
        there is a $*$-homomorphism $\msA(B)\to \mcR^\mcU$, where $\mcR^{\mcU}$
        is an ultrapower of the hyperfinite $II_1$ factor $\mcR$, and
    \item $\mcG(B)$ has a perfect commuting-operator strategy if and only if
        $\msA(B)$ has a tracial state. 
\end{enumerate} \end{theorem}
Although BCS algebras hadn't been invented at that point, parts (1) and (2) of
\Cref{thm:BCS_t-strats} were essentially proved in \cite{CM14}. Part (3)
and part (4) were first proved in \cite{KPS18} and \cite{PSSTW16} respectively for
synchronous games, with $\mcA(B)$ replaced by the synchronous algebra.
The statements for BCS games follow from \Cref{lem:padding} and \Cref{lem:synchgame}
in the next section. It's also possible to prove parts (3) and (4) directly using
the techniques in \cite{KPS18, PSSTW16}.  Background on $\mcR^{\mcU}$ can be found in \cite{CL15}.

A BCS with contexts is said to be \textbf{satisfiable} if it has a satisfying
assignment. \Cref{thm:BCS_t-strats} suggests the following definition:
\begin{definition}\label{def:sat_types}
    A BCS with contexts $B$ is: 
    \begin{enumerate}[(i)]
        \item \textbf{satisfiable} if there is a $*$-representation $\msA(B)\to \C$,
        \item \textbf{matrix-satisfiable} if there is a $*$-representation $\msA(B)\to M_d(\C)$
            for some $d \geq 1$, 
        \item \textbf{$\mcR^\mcU$-satisfiable} if there is a $*$-representation $\msA(B)\to \mcR^\mcU$,
        \item \textbf{tracially-satisfiable} if $\msA(B)$ has a tracial state, 
        \item \textbf{$C^*$-satisfiable} if $\msA(B)$ there is a $*$-representation $\msA(B) \to \mcB(H)$
            for some Hilbert space $H$, and
	\item \textbf{algebraically-satisfiable} if $1\neq 0\in \msA(B)$.
\end{enumerate} \end{definition}
Note that these types of satisfiability are organized from strongest to weakest.
Indeed, satisfiability obviously implies matrix-satisfiability. Since there are
homomorphisms $M_d \to \mcR^{\mcU}$ for all $d \geq 1$, matrix-satisfiability
implies $\mcR^{\mcU}$-satisfiability. Since $\mcR^{\mcU}$ has a tracial
state, $\mcR^{\mcU}$-satisfiability implies tracial-satisfiability. By the GNS
representation theorem, $\msA(B)$ has a $*$-representation $\msA(B) \to
\mcB(H)$ if and only if $\msA(B)$ has a state, so tracial-satisfiability
implies $C^*$-satisfiability. And finally, $C^*$-satisfiability implies
algebraically satisfiability.  Although $C^*$ and algebraic satisfiability
aren't linked to BCS games, $C^*$-satisfiable constraint systems do give rise
to quantum contextuality scenarios.

\subsection{Synchronous non-local games and algebras}\label{ss:synchronous}

A synchronous nonlocal game is a tuple $(\lambda, I, O)$, where $I$ and $O$ are
finite sets, and $\lambda : I \times I \times O \times O \to \{0,1\}$ is a
function satisfying the \textbf{synchronous condition} that $\lambda(i,i,a,b) =
0$ for all $i \in I$ and $a \neq b \in O$. In the game, Alice and Bob receive
questions $i,j \in I$ respectively, and reply with answers $a,b \in O$. They
win if $\lambda(i,j,a,b) = 1$. The synchronous condition states that if Alice
and Bob receive the same question, they must reply with the same answer. 
Given a synchronous game $\mcS = (\lambda, I, O)$, the 
\textbf{synchronous algebra} $\msA(\mcS)$ is the $*$-algebra generated by 
$e^i_a$, $a \in O$, $i \in I$, subject to the relations
\begin{enumerate}
\item $(e_a^i)^*=(e_a^i)^2=e_a^i$ for all $a\in O$, $i\in I$,
\item $\sum_{a\in A}e_a^i=1$ for all $i\in I$, and
\item $e_a^i e_b^j=0$ whenever $\lambda(i,j,a,b)=0$.
\end{enumerate}
Since $\lambda(i,i,a,b)=0$ for all $a \neq b$ by the synchronous condition,
the relations (3) imply in particular that $e^i_a e^i_b = 0$ for $a \neq b$.

A result of Fritz states that the class of synchronous algebras is the same as
the class of BCS algebras \cite{Fri18} (the proof is stated for the
$C^*$-enveloping algebras of synchronous and BCS algebras, but also applies to
the $*$-algebras). In the remainder of this section, we explain how to go back
and forth between the two directly. If $B = (X,\{(U_i,V_i)\}_{i=1}^{\ell})$ is
a BCS with contexts, the BCS nonlocal game $\mcG(B)$ is not necessarily a
synchronous game as described, since the number of possible answers to question
$i$ is $2^{|U_i|}$, which can depend on $i$.  However, we can make any BCS game
into a synchronous game using the following lemma: 
\begin{lemma}\label{lem:padding}
    If $B$ is a BCS with contexts, then there is a BCS with contexts $B'$
    such that $\msA(B) \iso \msA(B')$, and the contexts in $B'$ all have the
    same number of variables.
\end{lemma}
\begin{proof}
    Suppose $B = (X,\{(U_i,V_i)\}_{i=1}^{\ell})$, and let $\alpha = \min_i |U_i|$
    and $\beta = \max_i |U_i|$. Let $X' := X \cup \{y_1,\ldots,y_{\beta-\alpha}\}$, where
    $y_1,\ldots,y_{\beta-\alpha}$ are new variables, and let $U_i' := U_i \cup
    \{y_1,\ldots,y_{r_i}\}$, where $r_i = \beta - |U_i|$. Finally, let $B' = (X',
    \{U_i',V_i'\}_{i=1}^{\ell})$, where $V_i' = V_i \cup \{(y_j,
    \{1\})\}_{j=1}^{r_i}$. In other words, $B'$ is the constraint system where
    we have added variables to each context so that each context has the same
    number of variables, and added constraints forcing those variables to be
    $1$. It follows from \Cref{lem:presperconstraint} that $y_i = 1$ in
    $\msA(B')$ for all $1 \leq i \leq \beta-\alpha$, and hence $\msA(B) \iso \msA(B')$. 
\end{proof}
If $B'$ is constructed from $B$ as in the proof of \Cref{lem:padding}, then the
game $\mcG(B')$ is also essentially the same as $\mcG(B)$ --- the only difference
is that the players have to pad their answers with $1$'s so that all answer
strings have the same length. This version of BCS nonlocal games was first
considered in \cite{KPS18} (in a slightly different format).

Now suppose $B = (X,\{(U_i,V_i)\}_{i=1}^{\ell})$ is a constraint system with $n$
variables and $\ell$ contexts, each with exactly $m$ variables. For every $1
\leq i \leq \ell$, pick an order on $U_i$ so that we can think of the elements
of $\Z_2^m$ as assignments to $U_i$. Then $\mcG(B)$ is a synchronous game with
$I = \{1,\ldots,\ell\}$ and $O = \Z_2^{m}$. The rule function $\lambda$ for
$\mcG(B)$ is defined by $\lambda(i,j,a,b) = 1$ if and only if $a$ and $b$ are
satisfying assignments to $(U_i, V_i)$ and $(U_j,V_j)$ respectively and
$a|_{U_i \cap U_j} = b|_{U_i \cap U_j}$. 
\begin{lemma}\label{lem:synchgame}
    Suppose $B$ is a constraint system with contexts, in which every context
    has the same number of variables, so $\mcG(B)$ is a synchronous game.
    Suppose further that every variable of $B$ appears in some context $U_i$.
    Then the synchronous algebra $\msA(\mcG(B))$ is isomorphic to the BCS
    algebra $\msA(B)$. 
\end{lemma}
\Cref{lem:synchgame} is proved in \cite{Gol21} for linear systems $B$. The proof
for general constraint systems is similar; we supply the proof for the convenience
of the reader.
\begin{proof}
    Let $B = (X,\{(U_i,V_i)\}_{i=1}^{\ell})$. By definition, if $a$ is not a
    satisfying assignment to $(U_i,V_i)$, then the projection $\Pi_{U_i,a}$ is
    equal to $0$ in $\msA(B)$. Suppose $a$ and $b$ are assignments to $U_i$ and
    $U_j$ respectively for some $1 \leq i, j \leq \ell$. If $a(x) \neq b(x)$ for
    some $x \in U_i \cap U_j$, then 
    \begin{equation*}
        \Pi_{U_i,a} \Pi_{U_j,b} = \cdots \tfrac{1}{2}(1+a(x)x) \cdot \tfrac{1}{2}(1+b(x)x) \cdots
         = 0
    \end{equation*}
    in $\msA(B)$. Hence there is a homomorphism $\psi : \msA(\mcG(B)) \to \msA(B)$ sending
    $e^i_a \mapsto \Pi_{U_i,a}$.

    Conversely, for every $1 \leq i \leq \ell$ and $x \in U_i$, let
    \begin{equation*}
        \phi(x,i) = \sum_{a} a(x) \cdot e^i_a,
    \end{equation*}
    in $\msA(\mcG(B))$, where the sum is across assignments $a$ to $U_i$.
    Note that $\phi(x,i)$ is unitary and $\phi(x,i)^2 = 1$. If $x \in U_i \cap U_j$,
    then 
    \begin{align*}
        \phi(x,i) \phi(x,j) & = \left(\sum_{a} a(x) e^i_a \right) \cdot \left(\sum_b b(x) e^j_b \right)
            = \sum_{a|_{U_i \cap U_j} = b|_{U_i \cap U_j}} a(x)b(x) e^i_a e^j_b \\
            &  = \sum_{a|_{U_i \cap U_j} = b|_{U_i \cap U_j}} e^i_a e^j_b 
            = \left(\sum_{a} e^i_a \right) \cdot \left(\sum_b e^j_b \right) = 1.
    \end{align*}
    Since $\phi(x,i)$ and $\phi(x,j)$ are unitaries, we conclude that
    $\phi(x,i) = \phi(x,j)$, and we denote this element by $\phi(x)$. If $x, y \in U_i$
    for some $i$, then $\phi(x) = \phi(x,i)$ and $\phi(y) = \phi(y,i)$ commute.
    Since every variable $x \in X$ belongs to some context, there is a
    homomorphism $\phi : \msA_{con}(B) \to \msA(\mcG(B))$ sending $x \mapsto \phi(x)$. 
    Observe that
    \begin{equation*}
        \phi(\Pi_{U_i,a}) = \prod_{x \in U_i} \tfrac{1}{2} (1 + a(x) \phi(x))
                = \prod_{x \in U_i} \sum_{b : b(x)=a(x)} e^i_b = e^i_a.
    \end{equation*}
    If $a$ is not a satisfying assignment to $(U_i,V_i)$, then $\lambda(i,i,a,a)=0$
    and hence $e^i_a = e^i_a \cdot e^i_a = 0$, so the homomorphism $\phi$ descends
    to a homomorphism $\wtd{\phi} : \msA(B) \to \msA(\mcG(B))$. It is not hard
    to see that $\psi$ and $\wtd{\phi}$ are inverses. 
\end{proof}
\Cref{lem:synchgame} shows that, if every variable of $B$ is contained in some context,
then $\msA(B)$ is generated by the projections $e^i_a = \Pi_{U_i,a}$, subject to 
relations (1)-(3) for the synchronous algebra. 

If $x$ is a variable of $B = (X,\{(U_i,V_i)\}_{i=1}^{\ell})$ which does not
appear in any context, then $\msA(B) = \msA(B') \ast \C \Z_2$, where $B' = (X
\setminus \{x\},\{(U_i,V_i)\}_{i=1}^{\ell})$ is the constraint system with $x$
removed. We leave it as an exercise to the reader to show that such algebras
can also be expressed as synchronous algebras. 

Going the opposite direction, a convenient way to turn a synchronous algebra
into a BCS algebra is described in \cite{Pad22}: Suppose $\msS =
(I,O,\lambda)$ is a synchronous game. Let $\oneink \subseteq \Z_2^k$ be the
relation consisting of all vectors $(x_1,\ldots,x_k)$ such that exactly one of
$x_1,\ldots,x_k$ is equal to $-1$, and all others are equal to $1$. Let $B =
(X, V_1 \cup V_2)$, where
\begin{align*}
        X & = \{x_{ia} : i \in I, a \in O\},  \\
        V_1 & = \left\{x_{ia} \wedge x_{jb} = \FALSE\ :\ i,j \in I, a, b \in O \text{ such that } \lambda(i,j,a,b)=0\right\},  \\
        V_2 & = \left\{\left((x_{ia})_{a \in O}, \oneink\right) : i \in I \right\},
\end{align*} 
and $k= |O|$. Let $C_{ijab}$ be the constraint $x_{ia} \wedge x_{jb} = \FALSE$. The
point of the constraints in $V_1$ is that the only non-satisfying assignment
to $C_{ijab}$ is $\phi(x_{ia}) = \phi(x_{jb}) = -1$, and for this assignment
$\Pi_{C_{ijab},\phi} = \frac{(1-x_{ia})}{2} \frac{(1-x_{jb})}{2}$. By \Cref{lem:presperconstraint},
\begin{equation*}
    \frac{(1-x_{ia})}{2} \frac{(1-x_{jb})}{2} = 0
\end{equation*}
in $\msA(B)$ for all $i,j \in I$ and $a,b \in O$ with $\lambda(i,j,a,b)=0$. 
In particular, by the synchronous condition this relation holds when $i=j \in I$
and $a \neq b \in O$. When combined with the relations from $V_1$, the constraints
in $V_2$ imply that
\begin{equation*}
    \sum_{a} \frac{(1-x_{ia})}{2} = 1 
\end{equation*}
for all $i \in I$. Hence there is a homomorphism $\msA(\mcS) \to \msA(B)$ sending
$e^i_a \mapsto \frac{1-x_{ia}}{2}$, and this map is an isomorphism \cite{Pad22}.

\begin{example}
    Suppose $G$ and $H$ are graphs. The graph homomorphism game is the
    synchronous game $\mcS$ with $I = V(G)$, $O = V(H)$, and
    $\lambda(u,v,a,b)=0$ if and only if $u=v$ and $a \neq b$, or $u$ is
    adjacent to $v$ and $a$ is not adjacent to $b$. Let $B(G,H)$ be the
    constraint system with variables $X = \{x_{ua} : u \in V(G), a \in V(H)\}$
    and two types of constraints: $x_{ua} \wedge x_{vb} = \FALSE$ if $u=v \in
    V(G)$ and $a \neq b$, or $u$ and $v$ are adjacent in $G$ but $a$ and $b$
    are not adjacent in $V(H)$; and $((x_{ua})_{a \in V(H)} : \oneink)$ for $u
    \in V(G)$, where $k = |V(H)|$. Then $\msA(B(G,H))$ is the synchronous algebra of
    $\mcS$. 

    The constraint system $B(G,H)$ is satisfiable if and only if there is a
    homomorphism from $G$ to $H$. The chromatic number of $G$ is the smallest
    $n$ such that there is a homomorphism $G \to K_n$. In \cite{Cam07}, the
    \textbf{quantum chromatic number} of $G$ is defined to be the smallest $n$
    such that $B(G,H)$ has a matrix satisfying assignment.
    Quantum-approximable, commuting-operator, $C^*$, and algebraic chromatic
    numbers are defined similarly using $\mcR^{\mcU}$-satisfiability,
    tracial-satisfiability, $C^*$-satisfiability, and algebraic satisfiability
    respectively \cite{PSSTW16}, and these definitions (and follow-up work
    on synchronous algebras in \cite{HMPS19}) are part of the inspiration for
    Definition \ref{def:sat_types}.
\end{example}

\section{Separations between types of satisfiability}\label{sec:separations}

It is natural to ask whether any of the different types of satisfiability in
\Cref{def:sat_types} are equal. The Mermin-Peres magic square is an example of a
BCS which is matrix-satisfiable, but not satisfiable. Hence the corresponding
nonlocal game has a perfect quantum strategy, but no perfect classical strategy.
We can also separate the other types of satisfiability:
\begin{example}[\cite{Slof17}]\label{ex:connesembeddable}
    There is a linear system $B$ over $\Z_2$ which is $\mcR^\mcU$-satisfiable,
    but not matrix-satisfiable. The corresponding nonlocal game has a perfect
    quantum-approximable strategy, but no perfect quantum strategy.
\end{example}

\begin{example}[\cite{JNVWY22}]\label{ex:mip}
    The MIP*=RE theorem implies the existence of a synchronous algebra which
    has a tracial state, but does not have any homomorphisms to $\mcR^{\mcU}$.
    Hence there is a BCS which is tracially-satisfiable but not
    $\mcR^{\mcU}$-satisfiable. The corresponding nonlocal game has a perfect
    commuting-operator strategy, but no perfect quantum-approximable strategy.
\end{example}

\begin{example}[\cite{HMPS19}]\label{ex:notcstar}
    For any $G$ and $H$, $B(G,K_4)$ is algebraically satisfiable. However, there is
    a graph $G$ with $C^*$-chromatic number greater than $4$, meaning that $B(G,K_4)$
    is not $C^*$-satisfiable. We do not know of an operational interpretation of 
    this separation.
\end{example}

We can also give an example of a BCS which is $C^*$-satisfiable but not
tracially-satisfiable. For this, we need the following well-known lemma:
\begin{lemma}\label{lem:ms_relations}
    Let $B$ be the Mermin-Peres magic square in \Cref{ex:mp_square}. 
\begin{enumerate}[(a)]
\item In $\msA(B)$ we have $x_1 x_5 = - x_5 x_1$ and $x_2 x_4 = - x_4 x_2$, while
        $x_1$ and $x_5$ commute with $x_2$ and $x_4$.
\item If $X_i,Z_i$ are unitaries on a Hilbert space $H_i$, $i=1,2$ such that
        $X_i^2 = Z_i^2 = \Id$ and $X_i Z_i = - Z_i X_i$ for $i=1,2$, then there is a unique
        homomorphism
        \begin{equation*}
            \phi : \msA(B) \to \mcB(H_1 \otimes H_2)
        \end{equation*}
        such that $\phi(x_1) = Z_1 \otimes \Id$, $\phi(x_2) = \Id \otimes Z_2$,
        $\phi(x_4) = \Id \otimes X_2$, and $\phi(x_5) = X_1 \otimes \Id$. 
\end{enumerate}
\end{lemma}

\begin{prop}\label{prop:nottracial}
    There is a constraint system $B$ such that $\msA(B)$ is $C^*$-satisfiable,
    but not tracially-satisfiable.
\end{prop}
\begin{proof}
    Let $B_i = (X_i,V_i)$, $i=1,2$ be two copies of the Mermin-Peres magic
    square, in variables $X_i = \{x_{ij} : 1 \leq j \leq 9\}$ respectively.
    Consider the constraint system 
    \begin{equation*}
        B = (X_1 \cup X_2, V_1 \cup V_2 \cup \{x_{21} = x_{11} \wedge x_{12}\}).
    \end{equation*} 
    By \Cref{cor:relationpres} and \Cref{ex:zAND}, 
    \begin{equation*}
        \msA(B) = \msA(B_1) \ast \msA(B_2) / \langle x_{21} = \tfrac{1}{2} (1 + x_{11} + x_{12} - x_{11} x_{12}) \rangle.
    \end{equation*}
    Hence to give a representation $\phi : \msA(B) \to \mcB(H)$, it suffices to 
    give a pair of representations $\phi_i : \msA(B_i) \to \mcB(H)$, $i=1,2$
    such that
    \begin{equation*}
        \phi_2(x_{21}) = \tfrac{1}{2} \left(\Id + \phi_1(x_{11}) + \phi_1(x_{12})
         - \phi_1(x_{11} x_{12}) \right). 
    \end{equation*}
    Let $H_0 = \ell^2 \mbN$ be the countable dimensional Hilbert space, and let
    $X$ and $Z$ be the operators with block-matrix decomposition
    \begin{equation*}
        Z = \begin{pmatrix} \Id & 0 \\ 0 & -\Id \end{pmatrix}
        \text{ and } 
        X = \begin{pmatrix} 0 & \Id \\ \Id & 0 \end{pmatrix}
    \end{equation*}
    on $H_1 = H_0 \oplus H_0$. Let 
    \begin{equation*}
        W = \tfrac{1}{2} (\Id + Z \otimes \Id + \Id \otimes Z - Z \otimes Z)
            = \begin{pmatrix} 1 & 0 & 0 & 0 \\ 0 & 1 & 0 & 0 \\ 0 & 0 & 1 & 0 \\ 0 & 0 & 0 & -1 \end{pmatrix}
    \end{equation*}
    on $H_2 = H_1 \otimes H_1 = (H_0 \otimes H_0)^{\oplus 4}$. Now because $H_0$ is 
    countable dimensional, $H_0^{\oplus 3} \iso H_0$. Pick an isomorphism $T_0
    : H_0 \to H_0^{\oplus 3}$, and let $T$ be the operator with block-matrix decomposition
    \begin{equation*}
        T = \begin{pmatrix} 0 & T_0 \\ T_0^{-1} & 0 \end{pmatrix}
    \end{equation*}
    on $H_2 = (H_0 \otimes H_0)^{\oplus 3} \oplus H_0$. Observe that $W^2 = T^2 = 1$, and
    $W$ and $T$ anticommute. 

    Let $H = H_2 \otimes H_2$. Setting $Z_1 = Z$, $X_1 = X$, 
    $Z_2 = Z \otimes \Id_{H_2}$, and $X_2 = X \otimes \Id_{H_2}$ in
    \Cref{lem:ms_relations}, we see that there is a homomorphism 
    \begin{equation*}
        \phi_1 : \msA(B) \to \mcB(H)
    \end{equation*}
    with $\phi_1(x_{11}) = Z \otimes \Id_{H_1 \otimes H_2}$, $\phi_1(x_{12})
    = \Id_{H_1} \otimes Z \otimes \Id_{H_2}$, $\phi_1(x_{14}) = \Id_{H_1} \otimes
    X \otimes \Id_{H_2}$, and $\phi_1(x_{15}) = X \otimes \Id_{H_1 \otimes H_2}$. 
    Similarly, setting $Z_1 = Z_2 = W$ and $X_1 = X_2 = T$ in
    \Cref{lem:ms_relations}, we see that there is a homomorphism
    $\phi_2 : \msA(B) \to \mcB(H)$ with $\phi_2(x_{21}) = W \otimes \Id_{H_2}$, 
    $\phi_2(x_{22}) = \Id_{H_2} \otimes W$, $\phi_2(x_{24}) = \Id_{H_2} \otimes T$,
    and $\phi_2(x_{25}) = T \otimes \Id_{H_2}$. From this we see that
    \begin{align*}
         \tfrac{1}{2} (\Id + & \phi_1(x_{11}) + \phi_1(x_{12})
         - \phi_1(x_{11} x_{12}) ) \\
        & = \tfrac{1}{2} \left(\Id_{H_2} + Z \otimes \Id_{H_1} 
            + \Id_{H_1} \otimes Z - Z \otimes Z \right) \otimes \Id_{H_2} 
        = W \otimes \Id_{H_2} = \phi_2(x_{21}),
    \end{align*}
    and we conclude that $B$ is $C^*$-satisfiable.
        
    To see that $B$ is not tracially-satisfiable, suppose that $\tau$ is a tracial
    state on $\msA(B)$. If $a$ and $b$ are unitaries in $\msA(B)$ such that
    $a^2 = b^2 = 1$ and $ab = - ba$, then we have that $\tau(a) = \tau(bab) = -\tau(a)$,
    and hence $\tau(a)=0$. By \Cref{lem:ms_relations} again, $x_{11}$, $x_{12}$ and
    $x_{11} x_{12}$ anti-commute with $x_{15}$, $x_{14}$, and $x_{15}$ respectively, so
    $\tau(x_{11}) = \tau(x_{12}) = \tau(x_{11} x_{12}) = 0$. Hence
    \begin{equation*}
        \tau(x_{21}) = \tfrac{1}{2} \tau\left( 1 + x_{11} + x_{12} - x_{11} x_{12}\right)
            = \tfrac{1}{2}.
    \end{equation*}
    But $x_{21}$ anti-commutes with $x_{25}$, so $\tau(x_{21}) = 0$, a contradiction.
\end{proof}
Combined with \Cref{lem:synchgame}, \Cref{prop:nottracial} shows that there is
a synchronous algebra with a state, but no tracial state, answering a question
in \cite{Har21}. And while \Cref{prop:nottracial} is not directly relevant to
nonlocal games, it does imply that there are contextual operators which do not
admit a tracial state.

\section{Constraint languages and definability}\label{sec:define}

From the previous section, we know that all of the types of satisfiability in
\Cref{def:sat_types} are different. Given a constraint language $\mcL$, it
makes sense to ask: are any of the types of satisfiability in
\Cref{def:sat_types} equal for constraint systems over $\mcL$?  This question
was first asked by Atserias, Kolaitis, and Severini in \cite{AKS17}. In this section, we give
an overview of their results, and show how some proofs can be streamlined and
extended using BCS algebras. We then discuss how the problem of finding a
non-hyperlinear group fits into this picture.

Before proceeding, we first review some well-known constraint languages.
 Recall that, in this setting, a \textbf{term} is a boolean variable or a
constant from $\Z_2$.  A \textbf{literal} is a term, or the negation of a term.
A literal is said to be \textbf{negative} if it is the negation of a term, and
\textbf{positive} otherwise. A \textbf{clause} is a relation constructed by taking a 
disjunction of literals. A \textbf{conjunctive normal form (CNF)} formula is 
a conjunction (AND) of clauses, or equivalently a constraint system where are
all the relations are clauses. With this terminology:
\begin{enumerate}
\item A relation is \textbf{bijunctive} if it is the set of satisfying assignments
    of a CNF with at most two literals in every clause. The set of bijunctive
    relations is denoted by $2$SAT.
\item A relation is \textbf{Horn} if it is the set of satisfying assignments
    of a CNF in which all clauses have at most one positive term. The set of
    Horn relations is denoted by HORN.
\item A relation is \textbf{dual Horn} if it is the set of satisfying assignments
    of a CNF in which all clauses having at most one negative term. The set of
    dual Horn relations is denoted by DUAL-HORN.
\item A relation is \textbf{linear} (sometimes called \textbf{affine}) if it is
    the set of satisfying assignments of a linear system of equations. The set of
    linear relations is denoted by LIN.
\end{enumerate}
We say that a constraint language is bijunctive (resp.~Horn, dual Horn, or linear)
if every relation in the language is bijunctive (resp.~Horn, dual Horn, or linear),
or equivalently if it is a subset of $2$SAT (resp.~HORN, DUAL-HORN, or LIN). 

Note that we defined a $3$SAT instance to be a constraint system where every
relation is a clause with at most three literals, and we could similarly define
a 2SAT instance to be a constraint system where every relation is a clause with
at most two literals. This means that a constraint system over $2$SAT is somewhat
different from a $2$SAT instance, since a constraint system over $2$SAT can use
any bijunctive relation, not just clauses with at most two literals. 
However, if $B = (X, \{(U_i,V_i)\}_{i=1}^{\ell}$ and $B' = (X,
\{(U_i,V_i')\}_{i=1}^{\ell}$ are two constraint systems with the same contexts
and such that $V_i$ and $V_i'$ have the same satisfying assignments, then
$\msA(B) = \msA(B')$. Thus if we want to make a statement about $\msA(B)$ for
constraint systems $B$ over $2$SAT, we can assume that $B$ is a $2$SAT instance.
The same goes for constraint systems over HORN, DUAL-HORN, and so on, as well 
as constraint systems with contexts. 

Given a constraint language $\mcL$, let $\SAT(\mcL)$ denote the decision
problem of determining whether a constraint system $B$ over $\mcL$ has a
satisfying assignment. The classes of constraint languages defined in (1)-(4) are
sometimes called Schaefer classes (specifically, Schaefer classes with
constants), due to their appearance in Schaefer's dichotomy theorem \cite[Lemma
4.1]{Sch78}, which states that if a finite constraint language $\mcL$ is
bijunctive, Horn, dual Horn, or linear, then $\SAT(\mcL)$
is decidable in polynomial time, and otherwise $\SAT(\mcL)$ is NP complete.
Thus we normally think of the Schaefer classes as the tractable classes.  The
first observation of \cite{AKS17}, building on earlier work of \cite{Ji13}, is
that for every Schaefer class except LIN, all the types of satisfiability in
\Cref{def:sat_types} are the same:
\begin{lemma}[\cite{Ji13, AKS17}]\label{lem:noseps}
    If $B$ is a BCS with contexts over 2SAT, HORN, or DUAL-HORN, then $B$ is
    algebraically-satisfiable only if $B$ is satisfiable.
\end{lemma}
\begin{proof}
    Lemma 7 in \cite{AKS17} shows that if $B$ is a BCS with contexts over 2SAT,
    HORN, or DUAL-HORN, and $B$ is $C^*$-satisfiable, then $B$ is satisfiable.
    However, the proof shows that if $B$ is not satisfiable, then $1=0$ in
    $\msA(B)$, and hence $B$ is not algebraically satisfiable. 
\end{proof}

The main result of \cite{AKS17} is that if a constraint language $\mcL$ is not
in one of the classes in \Cref{lem:noseps}, then the first three types of
satisfiability in \Cref{def:sat_types} can be separated over $\mcL$.
\begin{theorem}[Theorem 3 of \cite{AKS17}]\label{thm:dichotomy1}
    If $\mcL$ is a constraint language which is not bijunctive, Horn, or
    dual Horn, then 
    \begin{enumerate}[(a)]
        \item there is a BCS with contexts over $\mcL$ which is 
            matrix-satisfiable but not satisfiable, and
        \item there is also a BCS with contexts over $\mcL$ which is
            $\mcR^{\mcU}$-satisfiable but not matrix-satisfiable.
    \end{enumerate}
\end{theorem}
Theorem 3 of \cite{AKS17} actually shows a weaker version of part (b), that $B'$
is $C^*$-satisfiable but not matrix-satisfiable. In this section, we give a
streamlined proof of \Cref{thm:dichotomy1} using BCS algebras, from which it's
clear that part (b) holds for $\mcR^{\mcU}$-satisfiability as well. As in \cite{AKS17},
we need to consider when one constraint language can be defined from another:
\begin{definition}\label{defn:definable} A relation $R \subseteq \{\pm 1\}^k$ is
    \textbf{definable} (or \textbf{$pp$-definable}) from a boolean constraint language
    $\mcL$ if there is a BCS $B$ over $\mcL$ with variables $\{x_1,\ldots,x_k\} \cup Y$
    such that $(a_1,\ldots,a_k)\in R$ if and only if there is a satisfying assignment $\phi$
    for $B$ with $\phi(x_i) = a_i$.

    We say that a BCS $B$ (resp.~BCS with contexts) is definable from $\mcL$ if
    every relation in $B$ is definable from $\mcL$. Similarly a constraint
    language $\mcL'$ is definable from $\mcL$ if every relation in $\mcL'$ is
    definable from $\mcL$.
\end{definition}

Then we have:
\begin{lemma}\label{lem:bcs_define}
    If a BCS with contexts $B$ is definable from the constraint language
    $\mcL$, then there exists a BCS with contexts $B'$ over $\mcL$ and
    $*$-homomorphisms 
    \begin{equation*} 
        \begin{tikzcd}
            \msA(B)\arrow[hookrightarrow, bend left]{r}{\iota} &
        \msA(B')\arrow[two heads, bend left]{l}{\pi} \end{tikzcd} 
    \end{equation*} 
    such that $\pi \circ \iota = \Id_{\msA(B)}$. 
\end{lemma}
\begin{proof}
    Let $B = (X, \{U_i, V_i\}_{i=1}^{\ell})$. For each BCS $(U_i, V_i)$, there is a
    BCS $(U_i \cup Y_i, W_i)$ over $\mcL$ such that $\phi$ is a satisfying
    assignment for $(U_i,V_i)$ if and only if there is a satisfying assignment
    $\wtd{\phi}$ for $(U_i \cup Y_i, W_i)$ with $\wtd{\phi}|_{U_i} = \phi$. Let $Y$
    be the disjoint union of sets $Y_i$, and consider the BCS with contexts $B'
    = (X \cup Y, \{(U_i \cup Y_i, W_i)\}_{i=1}^{\ell})$. Since $U_i$ is contained inside
    a context of $B'$ for all $1 \leq i \leq \ell$, there is a $*$-homomorphism 
    $\wtd{\iota} : \msA_{con}(B) \to \msA_{con}(B')$ sending $x \mapsto x$ for all $x \in X$.
    If $\phi$ is an assignment to $U_i \cup Y_i$, then 
    \begin{align*}
        \Pi_{U_i \cup Y_i,\wtd{\phi}} 
        =\prod_{x\in U_i} \left(\frac{1+\tilde{\phi}(x)x}{2}\right) \prod_{y \in Y_i} \left(\frac{1+\tilde{\phi}(y)y}{2}\right)
        =\wtd{\iota}(\Pi_{U_i,\phi}) \Pi_{Y_i, \wtd{\phi}},
    \end{align*}
    where $\phi = \wtd{\phi}|_{U_i}$, and $\Pi_{Y_i,\wtd{\phi}} := \prod_{y \in Y_i} \tfrac{1}{2} (1+\wtd{\phi}(y)y)$.
    For every assignment $\phi$ to $U_i$, we have 
    \begin{equation*}
        \sum_{\wtd{\phi}|_{U_i} = \phi} \Pi_{Y_i,\wtd{\phi}} = 1,
    \end{equation*}
    where the sum is over assignments $\wtd{\phi}$ to $U_i \cup Y_i$ with $\wtd{\phi}|_{U_i} = \phi$.
    If $\phi$ is a non-satisfying assignment for $(U_i,V_i)$, then every assignment $\wtd{\phi}$ to $U_i \cup Y_i$
    with $\wtd{\phi}|_{U_i} = \phi$ is a non-satisfying assignment to $(U_i \cup  Y_i, W_i)$. Therefore
    \begin{equation*}
        \wtd{\iota}(\Pi_{U_i,\phi}) = \sum_{\wtd{\phi}|_{U_i} = \phi} \wtd{\iota}(\Pi_{U_i,\phi}) \Pi_{Y_i,\wtd{\phi}}
            = \sum_{\wtd{\phi}|_{U_i} = \phi} \Pi_{U_i \cup Y_i, \wtd{\phi}}
    \end{equation*}
    vanishes in $\msA(B')$ for all non-satisfying assignments $\phi$ to $(U_i,V_i)$. We conclude that $\wtd{\iota}$ induces
    a homomorphism $\iota : \msA(B) \to \msA(B')$ sending $x \mapsto x$ for all $x \in X$.

    For the other direction, for each $1 \leq i \leq \ell$ and assignment
    $\phi$ to $U_i$, choose an assignment $h_{\phi}$ to $U_i \cup Y_i$, such
    that $h_{\phi}|_{U_i} = \phi$, and if $\phi$ is satisfying for $(U_i, V_i)$ then $h_{\phi}$ is satisfying
    for $(U_i \cup Y_i, W_i)$. Define $\wtd{\pi} : \msA_{con}(B') \to
    \msA_{con}(B)$ by $\wtd{\pi}(x) = x$ for all $x \in X$, and
    \begin{equation*}
        \wtd{\pi}(y) = \sum_{\phi} h_\phi(y) \Pi_{U_i,\phi}
    \end{equation*}
    if $y \in Y_i$, $1 \leq i \leq \ell$. Since $\wtd{\pi}(y)^* = \wtd{\pi}(y)$ and $\wtd{\pi}(y)^2 = 1$ for
    all $y \in Y$, and $\wtd{\pi}(z) \wtd{\pi}(w) = \wtd{\pi}(w) \wtd{\pi}(z)$
    commute for all $z,w \in U_i \cup Y_i$, the homomorphism $\wtd{\pi}$ is well-defined. 
    For any $a \in \{\pm 1\}$, we have that
    \begin{equation*}
        \tfrac{1}{2} (1 + a \pi(y)) = \sum_{\phi : h_{\phi}(y) = a} \Pi_{U_i,\phi}.
    \end{equation*}
    Suppose $\wtd{\phi}$ is an assignment to $U_i \cup Y_i$, and let $\phi_0 = \wtd{\phi}|_{U_i}$. Then 
    \begin{equation*}
        \wtd{\pi}(\Pi_{U_i \cup Y_i, \wtd{\phi}}) = \Pi_{U_i, \phi_0} \cdot \prod_{y \in Y_i} \left( \sum_{\phi : h_{\phi}(y) = \wtd{\phi}(y)} \Pi_{U_i,\phi} \right)
 = \begin{cases} \Pi_{U_i,\phi_0} & \wtd{\phi} = h_{\phi_0} \\
                                0 & \text{otherwise}
                            \end{cases}.
    \end{equation*}
    By construction, if $h_{\phi}$ is a non-satisfying assignment to $(U_i \cup Y_i, W_i)$, then 
    $\phi$ is a non-satisfying assignment to $(U_i, V_i)$. Hence if $\wtd{\phi}$ is a non-satisfying
    assignment to $(U_i \cup Y_i, W_i)$, then $\wtd{\pi}(\Pi_{U_i \cup Y_i, \wtd{\phi}})$ vanishes
    in $\msA(B)$. We conclude that $\wtd{\pi}$ induces a homomorphism $\pi : \msA(B') \to \msA(B)$
    with the property that $\pi \circ \iota = \Id_{\msA(B)}$. 
\end{proof}

To finish the proof of \Cref{thm:dichotomy1}, we need the following observation
from \cite{AKS17}:
\begin{lemma}[Lemma 8 of \cite{AKS17}]\label{lem:lindefinable}
    If $\mcL$ is a constraint language which is not bijunctive, Horn, or dual
    Horn, then LIN is definable from $\mcL$.
\end{lemma}
\begin{proof}
    A BCS without constants is a BCS $B = (X,\{C_i\}_{i=1}^m)$ in which the
    scope $S_i$ of every constraint $C_i$ is a sequence in $X$.  A relation $R$
    is definable without constants from a constraint language $\mcL$ if it is
    definable from $\mcL$ as in \Cref{defn:definable} via BCS $B$ without constants.
    A BCS is definable without constants from $\mcL$ if every relation is definable
    without constants. A constraint language is said to be $0$-valid (resp. $1$-valid)
    if every relation contains $(1,\ldots,1)$ (resp. $(-1,\ldots,-1)$). 
    Lemma 8 of \cite{AKS17} states that if $\mcL$ is not bijunctive, Horn, dual Horn,
    $0$-valid, or $1$-valid, then LIN is definable without constants from $\mcL$.

    Let $\mcL$ be a constraint language which is not bijunctive, Horn, or dual
    Horn, and let $\mcL'$ be the constraint language consisting of all
    relations which can be derived by inserting constants into the relations of
    $\mcL$. Since $\mcL\subseteq \mcL'$, $\mcL'$ is not bijunctive, Horn, or
    dual Horn. And $\mcL'$ is not $0$-valid or $1$-valid, since it is not hard to
    see that then $\mcL$ would be Horn or dual Horn. So LIN is definable from
    $\mcL'$ without constants, and hence LIN is definable from $\mcL$ with constants.
\end{proof}

\begin{proof}[Proof of \Cref{thm:dichotomy1}]
    Suppose $\mcL$ is not bijunctive, Horn, or dual Horn, 
    and let $B$ be the Mermin-Peres magic square from \Cref{ex:mp_square}.
    Since $B$ is a constraint system over LIN, $B$ is definable from $\mcL$ by
    \Cref{lem:lindefinable}. By \Cref{lem:bcs_define}, there is a BCS with
    contexts $B'$ over $\mcL$ such that there are homomorphisms $\iota :
    \msA(B) \to \msA(B')$ and $\pi : \msA(B') \to \msA(B)$. Thus for any
    $*$-algebra $\mcC$, there is a homomorphism $\msA(B) \to \mcC$ if and only
    if there is a homomorphism $\msA(B') \to \mcC$.  Since $B$ is
    matrix-satisfiable but not satisfiable, so is $B'$. Thus there is a
    constraint system over $\mcL$ which is matrix-satisfiable but not
    satisfiable. 

    If we replace the Mermin-Peres square with the constraint system from 
    \Cref{ex:connesembeddable} which is $\mcR^{\mcU}$-satisfiable but not
    matrix-satisfiable, then the same argument shows that there is a
    constraint system over $\mcL$ which is $\mcR^{\mcU}$-satisfiable but not
    matrix-satisfiable.
\end{proof}
More generally, suppose that there is a constraint system $B$ with contexts
over LIN such that there is a $*$-homomorphism $\msA(B) \to \mcC$ for some
$*$-algebra $\mcC$, but no $*$-homomorphism $\msA(B) \to \mcC'$ for any $\mcC'$
in some given family of $*$-algebras. Then \Cref{lem:bcs_define} and
\Cref{lem:lindefinable} imply that for every constraint language $\mcL$ which
is not bijunctive, Horn, or dual Horn, there is a
constraint system with contexts $B'$ over $\mcL$ such that there is a
homomorphism $\msA(B') \to \mcC$, but no $*$-homomorphism to any $\mcC'$.  In
particular:
\begin{prop}
    If there is a constraint system with contexts over LIN which is
    tracially-satisfiable but not $\mcR^{\mcU}$-satisfiable, then for every
    constraint language $\mcL$ which is not bijunctive, Horn, or dual Horn,
    there is a constraint system over $\mcL$ which is tracially-satisfiable but
    not $\mcR^{\mcU}$-satisfiable. 
\end{prop}

\begin{lemma}[\cite{CLS17, SV19}]\label{lem:lin_satisfiable}
    Let $B$ be the BCS with contexts corresponding to a linear system
    $Ax=b$ over $\Z_2$. Let $J$ be the central element of the solution group 
    $\Gamma(A,b)$ from \Cref{ex:linear}. Then:
    \begin{enumerate}[(a)]
        \item $B$ is tracially satisfiable if and only if $J \neq 1$ in $\Gamma(A,b)$.
        \item $B$ is $\mcR^{\mcU}$-satisfiable if and only if $J$ is non-trivial
            in approximate representations of $\Gamma(A,b)$. 
    \end{enumerate}
\end{lemma}
Part (a) of \Cref{lem:lin_satisfiable} is proved in \cite{CLS17}, and part (b) is
proved in \cite{SV19}.  Here an element $w$ of a finitely presented group $G =
\langle S : R \rangle$ is said to be non-trivial in approximate representations
of $G$ if and only if for any representative word $w_0$ for $w$ and $\eps > 0$
there is a representation $\phi : \mcF(S) \to U(M_n \C)$ such that
$\|\phi(r)-\Id\|_2 \leq \eps$ for all $r \in S$ and $\|\phi(w_0)-\Id\|_2 \geq
1/4$, where $\mcF(S)$ is the free group generated by $S$ and $\|\cdot\|_2$ is
the normalized Frobenius norm. The constant $1/4$ here is somewhat arbitrary,
and can be replaced with any constant in $(0,\sqrt{2})$. Equivalently $w$ is
non-trivial in approximate representations if and only if there is a
homomorphism $\psi : \C G \to \mcR^{\mcU}$ such that $\psi(w) \neq 1$ (see,
e.g., \cite{CL15} for more background). 

Recall that a finitely-presented group $G$ is \textbf{hyperlinear} if every
non-trivial element is non-trivial in approximate representations. 
\begin{prop}
    There is a constraint system $B$ with contexts over LIN such that $B$ is
    tracially-satisfiable but not $\mcR^{\mcU}$-satisfiable if and only if
    there is a non-hyperlinear finitely presented group.
\end{prop}
\begin{proof}
    Suppose $B$ is a constraint system with contexts over LIN such that $B$ is
    tracially satisfiable but not $\mcR^{\mcU}$-satisfiable. As discussed just
    after the definition of LIN, we can assume that all the relations of $B$
    are linear relations. By adding ancillary variables and relations as in
    \Cref{ex:commute2}, we can further assume that $B$ is a linear system $Ax=b$,
    regarded as a BCS with contexts in the default way. By
    \Cref{lem:lin_satisfiable}, the element $J$ is non-trivial in
    $\Gamma(A,b)$, but trivial in approximate representations. Hence $\Gamma(A,b)$
    is non-hyperlinear.

    Conversely, suppose that $G = \langle S : R \rangle$ is a non-hyperlinear
    finitely-presented group, and let $w$ be a non-trivial element of $G$ which
    is trivial in approximate representations. We can assume that $w$ has infinite
    order; indeed, if $w$ does not have infinite order, then we can replace $G$
    with $G \ast G$ and $w$ with $\iota_1(w) \iota_2(w)$, where $\iota_1,
    \iota_2 : G \to G \ast G$ are the inclusions into the first and second
    factor respectively. Let $J$ be the generator of $\Z_2$ in the product $G
    \times \Z_2$. Since $x$ has infinite order, the subgroups $\langle w \rangle$
    and $\langle w J \rangle$ of $G$ are both isomorphic to $\Z$. Let 
    \begin{equation*}
        G' = \langle S \cup \{J,t\} : R \cup \{sJ=Js : s \in S\} \cup \{J^2=1, twt^{-1} = wJ\} \rangle
    \end{equation*}
    be the HNN extension of $G \times \Z_2$ by the map $w \mapsto wJ$. Then the homomorphism
    $G \times \Z_2$ into $G'$ is injective, so $J$ is non-trivial in $G'$. But if
    $\phi : \C G' \to \mcR^{\mcU}$ is a homomorphism, then $\phi(J) = \phi(t w t^{-1} w^{-1}) = 1$,
    since $\phi(w) = 1$. Hence $J$ is trivial in approximate representations of $G'$. 

    By \cite[Theorem 3.1]{Slof19}, there is a linear system $Ax=b$ and an
    injective homomorphism $\psi : G' \mapsto \Gamma(A,b)$ such that $\psi(J) = J$. 
    So the element $J \in \Gamma(A,b)$ is non-trivial, but trivial in approximate
    representations. Hence $Ax=b$ is tracially satisfiable, but not
    $\mcR^{\mcU}$-satisfiable.
\end{proof}
\Cref{thm:dichotomy1} cannot hold as stated for separations between between
tracial satisfiability, $C^*$-satisfiability, and algebraic satisfiability, 
since tracial satisfiability, $C^*$-satisfiability, and algebraic
satisfiability are all equivalent for constraint systems over LIN:
\begin{cor}\label{cor:nofun}
    If $B$ is a BCS with contexts over LIN, then $B$ is tracially satisfiable if and only
    if $B$ is algebraically satisfiable.
\end{cor}
\begin{proof}
    As in the last proof, we can assume that $B$ is a linear system $Ax=b$.  If
    $B$ is algebraically satisfiable, then $1 \neq 0$ in $\msA(B) = \C
    \Gamma(A,b) / \langle J=-1\rangle$. Hence $J \neq 1$ in $\Gamma(A,b)$, so 
    $B$ is tracially satisfiable by \Cref{lem:lin_satisfiable}.
\end{proof}
Instead we get:
\begin{thm}
    If $\mcL$ is a constraint language which is not bijunctive, Horn, dual Horn, or linear,
    then
    \begin{enumerate}[(a)]
        \item there is a BCS with contexts over $\mcL$ which is $C^*$-satisfiable but not 
            tracially satisfiable, and
        \item there is a BCS with contexts over $\mcL$ which is algebraically satisfiable
            but not $C^*$-satisfiable.
    \end{enumerate}
\end{thm}
\begin{proof}
    Suppose $\mcL$ is not bijunctive, Horn, dual Horn, or linear. By Theorem 3 of \cite{Sch78}, 
    every relation is definable from $\mcL$. Thus (a) and (b) follow from 
    \Cref{lem:bcs_define}, \Cref{prop:nottracial}, and \Cref{ex:notcstar}. 
\end{proof}

As mentioned in \Cref{ex:mip}, the MIP*=RE theorem gives a BCS $B$ which is
tracially-satisfiable but not $\mcR^{\mcU}$-satisfiable. One idea for
constructing a non-hyperlinear group that's seen a lot of investigation is to
try to find a linear system $B'$ which is algebraically satisfiable, and for
which there is a homomorphism $\msA(B) \to \msA(B')$. The existence of
constraint systems which are algebraically satisfiable but not tracially
satisfiable (from \Cref{prop:nottracial} and \Cref{ex:notcstar}) shows that
there is no generic procedure to construct such a $B'$ for any algebraically
satisfiable (or even $C^*$-satisfiable) BCS $B$. Indeed, if there is a
homomorphism $\msA(B) \to \msA(B')$ such that $B'$ is linear and algebraically
satisfiable, then by \Cref{cor:nofun} both $B$ and $B'$ are tracially
satisfiable. Thus for this approach to constructing a non-hyperlinear group to
work, the construction of $B'$ would have to use some particular feature of
$B$, such as the fact that it is tracially satisfiable.

We can also look at analogues of $\SAT(\mcL)$ for other types of
satisfiability.  Given a constraint language $\mcL$, let $\SAT_t(\mcL)$ denote
the decision problem of determining whether a BCS with contexts over $\mcL$ is
$t$-satisfiable, where $t$ is one of matrix, $\mcR^{\mcU}$, tracially, $C^*$,
or algebraically.  If $\mcL$ is finite, then the constraint system $B'$ in
\Cref{lem:bcs_define} is computable from $B$, and hence $\SAT_t(\mcL')$ is
Turing reducible to $\SAT_t(\mcL)$. When combined with \Cref{lem:lindefinable},
this gives dichotomy theorems for all the types of satisfiability in
\Cref{def:sat_types}:
\begin{theorem}[\cite{AKS17}]\label{thm:dichotomy2}
    Suppose $\mcL$ is a finite constraint language. If $\mcL$ is bijunctive,
    Horn, or dual Horn, then $\SAT_t(\mcL)$ is solvable in
    polynomial time by \Cref{lem:noseps}. Otherwise:
    \begin{enumerate}[(a)]
        \item $\SAT_{matrix}(\mcL)$ is undecidable, and contained in RE.
        \item $\SAT_{\mcR^{\mcU}}(\mcL)$ is coRE-hard, and contained in $\Pi_2^0$.
        \item $\SAT_{tracial}(\mcL)$ is coRE-complete.
        \item $\SAT_{C^*}(\mcL)$ is coRE-complete.
        \item $\SAT_{algebraic}(\mcL)$ is coRE-complete.
    \end{enumerate}
\end{theorem}
Only part (d) of \Cref{thm:dichotomy2} is proved in \cite{AKS17}, but the other
parts can be proved in the same way.
\begin{proof}
    If $\mcL$ is bijunctive, Horn, or dual Horn, then
    $\SAT_t(\mcL)$ is equivalent to $\SAT(\mcL)$ for all $t$ by
    \Cref{lem:noseps}, and thus all of these problems are solvable in
    polynomial time. 

    Suppose $\mcL$ is not bijunctive, Horn, or dual Horn.  By
    Lemma \ref{lem:lindefinable}, LIN is definable from $\mcL$, so
    $\SAT_t(\LIN)$ reduces to $\SAT_t(\mcL)$. 

    For part (a), \cite[Corollary 4]{Slof19} states that it is undecidable to
    determine if a linear system is matrix-satisfiable, so
    $\SAT_{matrix}(\mcL)$ is also undecidable. On the other hand,
    $\SAT_{matrix}(\mcL)$ is clearly contained in RE, since we can search
    through finite-dimensional strategies for a satisfying assignment. 

    For part (b), $B = (X, \{U_i, V_i\}_{i=1}^{\ell})$ has an
    $\mcR^{\mcU}$-satisfying assignment if and only if for all $\eps > 0$,
    there is a finite-dimensional strategy for $\mcG(B)$ which wins on every
    pair of inputs with probability $1-\eps$.\footnote{This condition can also
    be stated algebraically by saying that there is an $\eps$-representation of
    $\msA(B)$ for all $\eps > 0$. See \cite{Pad22} for details.} For any fixed
    $\eps > 0$, the existence of such a finite-dimensional strategy is in RE,
    so $\SAT_{\mcR^{\mcU}}(\mcL)$ is contained in $\Pi_2^0$. By \cite[Theorem
    2]{Slof17}, $\SAT_{\mcR^{\mcU}}(\LIN)$ is coRE-hard, so
    $\SAT_{\mcR^{\mcU}}(\mcL)$ is also coRE-hard. 

    For parts (c)-(e), $\SAT_{tracial}(\LIN)$ is coRE-hard by \cite[Corollary
    3.3]{Slof19}. By \Cref{cor:nofun}, $\SAT_{tracial}(\LIN)$, $\SAT_{C^*}(\LIN)$,
    and $\SAT_{algebraic}(\LIN)$ are all equivalent, so these problems are all
    coRE-hard. For upper bounds, $\SAT_{algebraic}(\mcL)$ is in coRE, since 
    when $0=1$ in $\msA(B)$ it is possible to demonstrate this with a sequence
    of relations. Similarly, as a semi-pre-$C^*$-algebra, $\msA(B)$ has a
    homomorphism to a $C^*$-algebra if and only if $-1$ is not positive \cite{HM04},
    and has a tracial state if and only if $-1$ is not a sum of hermitian squares
    and commutators \cite{KS08} (see also \cite{Oza13}). Hence
    $\SAT_{tracial}(\mcL)$, $\SAT_{C^*}(\mcL)$, and $\SAT_{algebraic}(\mcL)$ are all
    coRE-complete. 
\end{proof}
Note that the upper and lower bounds on $\SAT_{matrix}(\mcL)$ and
$\SAT_{\mcR^{\mcU}}(\mcL)$ in parts (a) and (b) of \Cref{thm:dichotomy2} are
not tight. However, by the MIP*=RE theorem determining whether an arbitrary
BCS is matrix-satisfiable is RE-complete. And using the MIP*=RE theorem, Mousavi
Nezhadi, and Yuen have shown that the problem of determining whether a BCS is
$\mcR^{\mcU}$-satisfiable is $\Pi^0_2$-complete \cite{MNY22}. Thus it seems
reasonable that:
\begin{conjecture}
    In parts (a) and (b) of \Cref{thm:dichotomy2}, $\SAT_{matrix}(\mcL)$ is
    RE-complete, and $\SAT_{\mcR^{\mcU}}(\mcL)$ is $\Pi^0_2$-complete.
\end{conjecture}
The challenge is showing that these lower bounds hold for $\SAT_{matrix}(\LIN)$ and
$\SAT_{\mcR^{\mcU}}(\LIN)$. Proving that $\SAT_{\mcR^{\mcU}}(\LIN)$ is
$\Pi^0_2$-complete (or even just that it contains RE) would imply the
existence of a non-hyperlinear group. 

\bibliography{bibfile} \bibliographystyle{alpha}

\end{document}

%% file: math_symbs_header.tex
\usepackage{amsmath,amsthm,amsfonts,amssymb,latexsym,verbatim,bbm}
\usepackage{subcaption}
\usepackage[shortlabels]{enumitem}
\usepackage{afterpage}
\usepackage{mathtools}
\usepackage{mathrsfs}
\usepackage{hyperref}
\usepackage[capitalize]{cleveref}

\numberwithin{equation}{section}
\theoremstyle{definition}
\newtheorem{theorem}{Theorem}[section]
\newtheorem{thm}[theorem]{Theorem}
\newtheorem{lemma}[theorem]{Lemma}

\newtheorem{prop}[theorem]{Proposition}
\newtheorem{example}[theorem]{Example}
\newtheorem{definition}[theorem]{Definition}
\newtheorem{defn}[theorem]{Definition}

\newtheorem{cor}[theorem]{Corollary}
\newtheorem{conjecture}[theorem]{Conjecture}

\newcommand{\iso}{\cong}

\newcommand{\R}{\mathbb{R}}
\newcommand{\C}{\mathbb{C}}
\newcommand{\Z}{\mathbb{Z}}

\newcommand{\eps}{\epsilon}

\newcommand{\Id}{\mathbbm{1}}

\newcommand{\mcA}{\mathcal{A}}
\newcommand{\mcB}{\mathcal{B}}
\newcommand{\mcC}{\mathcal{C}}

\newcommand{\mcF}{\mathcal{F}}
\newcommand{\mcG}{\mathcal{G}}

\newcommand{\mcJ}{\mathcal{J}}

\newcommand{\mcL}{\mathcal{L}}

\newcommand{\mcR}{\mathcal{R}}
\newcommand{\mcS}{\mathcal{S}}

\newcommand{\mcU}{\mathcal{U}}

\newcommand{\mbN}{\mathbb{N}}

\newcommand{\msA}{\mathscr{A}}
\newcommand{\msB}{\mathscr{B}}

\newcommand{\msS}{\mathscr{S}}

\newcommand{\End}{\text{End}}